\documentclass[letterpaper]{article}
\usepackage{aaai2027}
\nocopyright
\usepackage[hyphens]{url}
\usepackage{graphicx}
\def\UrlFont{\rm}
\usepackage{natbib}
\usepackage{caption}
\usepackage{amsmath}
\usepackage{amssymb}
\usepackage{amsthm}
\usepackage{algpseudocode}
\usepackage{enumerate}
\usepackage{siunitx}
\usepackage{pbox}
\usepackage{subcaption}
\usepackage{tikz}
\usepackage[linesnumbered,ruled]{algorithm2e}
\usepackage{amsfonts}
\usepackage{stackengine}
\usepackage{booktabs}
\usepackage{array}
\usepackage{multirow}
\usepackage{pifont}
\usepackage{color}
\usepackage{pdflscape}
\usepackage{float}
\usepackage{rotating}

\newcommand{\revblue}[1]{#1}
\newtheorem{theorem}{Theorem}
\newtheorem{assumption}{Assumption}
\newtheorem{remark}{Remark}

\newtheorem{lemma}{Lemma}
\newtheorem{definition}{Definition}

\begin{document}\title{Uncertainty Quantification via Invariant-Measure Conformal Prediction}
\author{
	Mohammadhossein~Bakhtiaridoust,
	Dominik~Baumann,
	and Shankar~Deka
}
\affiliations{
	Department of Electrical Engineering and Automation (EEA), Aalto University, Finland\\
	mohammad.bakhtiaridoust@aalto.fi, shankar.deka@aalto.fi, dominik.baumann@aalto.fi
}
\maketitle

\begin{abstract}
Uncertainty quantification for learned stochastic dynamical systems is essential in safety-critical tasks such as control and monitoring. Standard conformal prediction provides finite-sample coverage guarantees under exchangeability, but this assumption is typically violated in dynamical systems because trajectory data are temporally dependent, state distributions evolve, and recursive prediction errors accumulate. This paper proposes an invariant-measure conformal prediction (imCP) framework that calibrates uncertainty using independent samples from an invariant measure of the Markov process induced by the dynamics. This aligns calibration with the stationary operating regime and restores the statistical symmetry needed for rolling one-step split conformal guarantees. For recursive multi-step prediction, imCP combines conformal calibration with Lipschitz error propagation through the learned predictor to obtain explicit horizon-dependent bounds.These pre-deployment uncertainty tubes are suitable for rolling and receding-horizon applications, such as self-triggered control and fault detection, where uncertainty bounds must be computed before future residuals are observed. Numerical experiments show that imCP yields reliable bounds, while non-invariant calibration can become misaligned during deployment.
\end{abstract}
\definecolor{limegreen}{rgb}{0.2, 0.8, 0.2}
	\definecolor{forestgreen}{rgb}{0.13, 0.55, 0.13}
	\definecolor{greenhtml}{rgb}{0.0, 0.5, 0.0}
	\newcommand{\greencheck}{\textcolor{forestgreen}{\checkmark}}
	\newcommand{\redcross}{\textcolor{red}{\(\times\)}}

\section{Introduction}

In many applications, predictive models are expected not only to provide accurate point predictions, but also to quantify the uncertainty surrounding their outputs. For example, uncertainty estimates can support decisions in demand forecasting, resource allocation, control, and monitoring, where the reliability of a prediction is often as important as its nominal value \cite{angelopoulos2022gentle,lindemann2023safe}. Conformal prediction is a model-agnostic framework that is particularly well suited for data-driven and neural-network predictors, since it converts point predictions into prediction sets with finite-sample coverage guarantees \cite{shafer2008tutorial,angelopoulos2025foundations}. The classical conformal prediction framework was first introduced in \cite{shafer2008tutorial}; later, split conformal prediction improved its computational practicality by separating model fitting from calibration \cite{vovk2012conditional}. This efficient formulation has motivated many extensions, including conformalized quantile regression, time-series conformal methods, and adaptive procedures for different deployment regimes \cite{romano2019cqr,stankeviciute2021conformal,xu2021enbpi,zaffran2022adaptive}.

The key requirement behind these guarantees is that the calibration scores are statistically representative of the future test score. In the standard theory, this is ensured by exchangeability: informally, the calibration and test examples can be permuted without changing their joint distribution. When this calibration--deployment alignment is lost, the empirical conformal quantile may no longer characterize the future prediction error, and the resulting intervals can lose their intended coverage interpretation. Such failures are common in modern deployment settings, where temporal dependence, distribution shift, covariate drift, or adaptive data collection mechanisms can break exchangeability \cite{barber2023beyond,gibbs2024online}.

Dynamical systems make this issue particularly subtle. Samples collected along a trajectory are temporally dependent, the state distribution evolves under the dynamics, and multi-step prediction errors can accumulate through recursive application of the learned model. Therefore, a naive use of trajectory residuals as if they were exchangeable calibration samples is generally unjustified. Existing conformal methods address related aspects of this problem from different perspectives. Weighted conformal prediction corrects certain calibration--deployment mismatches through likelihood-ratio weighting when reliable density-ratio information is available \cite{tibshirani2019covariate}. Adaptive conformal inference and related online procedures update quantiles using past miscoverage indicators to maintain long-run coverage under changing environments \cite{gibbs2021aci}. Time-series conformal methods such as CF-RNN and EnbPI construct uncertainty intervals under assumptions such as exchangeability across sampled trajectories, mixing conditions, or approximate validity of residual processes \cite{stankeviciute2021conformal,xu2021enbpi}.

This paper addresses a complementary question: rather than correcting an arbitrary calibration--deployment mismatch after it occurs, can the calibration distribution be chosen so that it matches the rolling operating regime of the dynamical system from the outset? Existing approaches that adapt to mismatch or distribution shift are highly useful but their correction mechanisms typically require observing shifted data, past errors, or miscoverage events before the uncertainty level is adjusted. This reactive nature is less suitable for applications where uncertainty bounds must be available before the future trajectory and residuals are realized. Such a prospective requirement is common in dynamical-system applications involving control and monitoring. For example, in self-triggered control over shared communication networks, an agent must decide at the current communication instant when it will next require access to the network, based on a prediction of how its uncertainty evolves over future steps \cite{6425820}. Similarly, observer design and fault detection require nominal residual envelopes against which future deviations can be tested \cite{bakhtiaridoust2023data}. These applications also share a rolling or receding-horizon structure: the predictor is repeatedly initialized from the current measured or estimated state, and the uncertainty bounds must remain meaningful across repeated one-step or short-horizon prediction problems. Thus, one needs an uncertainty tube that can be computed before deployment or at the decision time, rather than a bound that is corrected only after future errors are observed.

In this work, we leverage invariant measures to construct prospective rolling uncertainty bounds for dynamical systems, which we refer to as invariant-measure conformal prediction (imCP). We study both a one-step rolling setting, where the predictor is repeatedly initialized from the current true state, and a recursive multi-step setting, where prediction errors propagate through the learned dynamics. In the multi-step case, imCP calibrates a single local one-step score in the invariant operating regime and propagates this uncertainty through the learned predictor using a Lipschitz error-propagation argument, yielding explicit horizon-dependent uncertainty radii. This makes imCP distinct from trajectory-level conformal methods, which often calibrate conservative worst-case trajectory scores \cite{cleaveland2024conformal}, and from fixed-horizon methods, which typically require separate quantiles for different horizons \cite{lindemann2023safe}. Fig.~\ref{fig:opening_invariant_calibration} illustrates how calibration from the invariant measure, rather than from a non-invariant distribution, affects the resulting coverage guarantees.

\begin{figure}[t]
	\centering
	\includegraphics[width=\columnwidth]{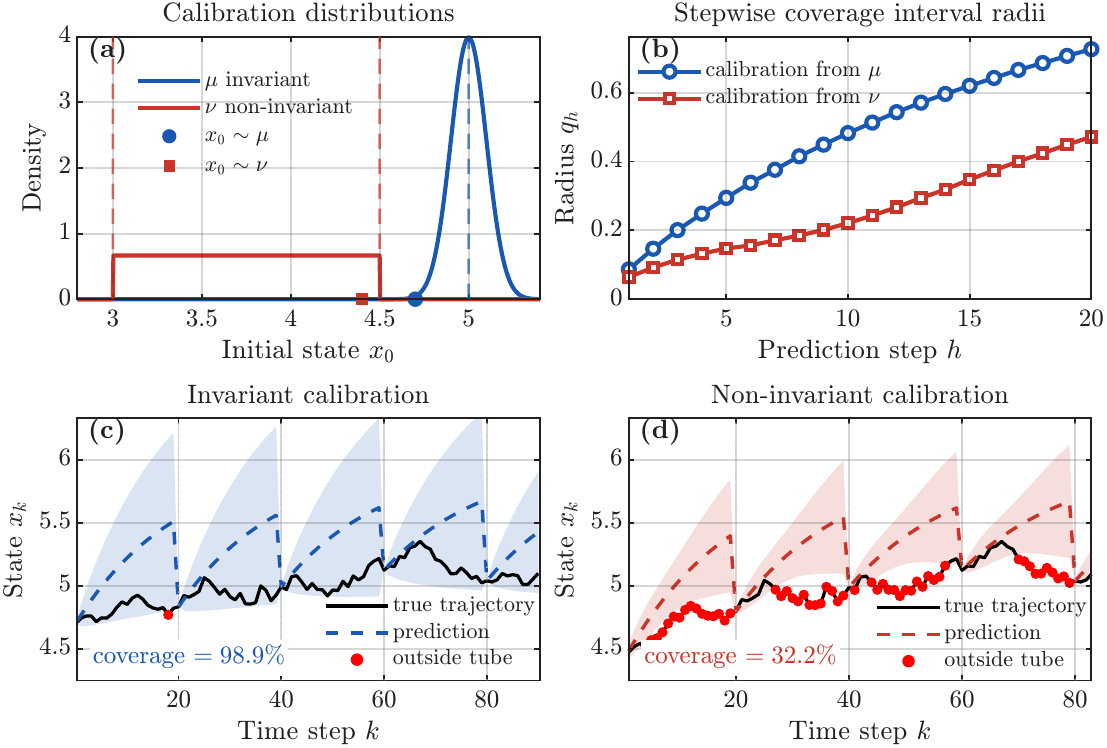}
	\caption{Motivating illustration of imCP calibration for
	 rolling conformal prediction. Calibration from the invariant measure keeps the calibration and deployment score laws aligned.}

	\label{fig:opening_invariant_calibration}
\end{figure}

Table 1 in Appendix B summarizes how imCP differs from existing conformal prediction approaches along several dimensions that are central to stochastic dynamical systems. 


The main contributions of this paper are summarized as follows:
\begin{itemize}
    \item We introduce an invariant-measure-based calibration framework for conformal prediction in stochastic dynamical systems, where calibration is aligned with the stationary operating distribution of the system.
    \item We establish a rolling one-step conformal prediction guarantee under independent calibration samples drawn from the invariant measure of the induced Markov process.
    \item We develop a recursive multi-step prediction formulation and derive explicit horizon-dependent uncertainty bounds with pathwise coverage guarantees by combining one-step conformal calibration with Lipschitz error propagation.
\end{itemize}

The remainder of the paper is organized as follows. Section II presents the preliminaries and problem formulation. Section III introduces imCP and its theoretical results. Section IV provides simulation studies demonstrating the effectiveness of the approach. Finally, Section V concludes the paper and discusses future research directions.


\section{Preliminaries and Problem Formulation}\label{sec:preliminaries}
\label{sec:preliminaries_problem}

This section introduces the dynamical-system setting, the score function used for calibration, and the uncertainty-quantification problems that imCP addresses. A full review of standard split conformal prediction, including the generic supervised notation and finite-sample coverage statement, is provided in Appendix A.

\subsection{Score function}
\label{subsec:prelim_split_cp}

Given an approximate dynamical predictor $\hat f:\mathcal{X}\to\mathcal{X}$, we use a nonconformity score
\begin{equation}
	s_{\hat f} : \mathcal{X} \times \mathcal{X} \rightarrow \mathbb{R}_{\ge 0},
	\label{eq:score_function}
\end{equation}
where $s_{\hat f}(x,\xi)$ measures how well a candidate successor $\xi$ agrees with the prediction $\hat f(x)$. In the experiments we use Euclidean residual scores, but the methodology only requires a fixed score function to construct and evaluate conformal prediction sets.

\subsection{Problem formulation}
\label{subsec:dynamical_system_setup}
\label{subsec:problem_formulation}

We consider a discrete-time stochastic dynamical system
\begin{equation}
	x_{k+1} = f(x_k)+n_k,
	\label{eq:dyn_true}
\end{equation}
where \(x_k \in \mathcal{X}\subseteq \mathbb{R}^{d}\), \(f:\mathcal{X}\to\mathcal{X}\) is the nominal, generally unknown, one-step evolution map, and \(\{n_k\}\) is a noise sequence defined on a common probability space. We assume that the noise sequence is independent of the past given the current state, so the system induces a Markov transition kernel \(P\) on \((\mathcal{X},\mathcal{B}(\mathcal{X}))\),
\begin{equation}
	P(x,A) := \mathbb{P}\bigl(f(x)+n_k \in A \,\big|\, x_k=x\bigr),
	\qquad A\in \mathcal{B}(\mathcal{X}).
	\label{eq:dyn_kernel}
\end{equation}
In practice, the true dynamics are rarely known exactly, and one works with an approximate predictor
\begin{equation}
\hat f:\mathcal{X}\to\mathcal{X},
\label{eq:dyn_model}
\end{equation}
constructed using different identification methods \cite{ljung1998system}.

Given the current state \(x_k\), the objective is to predict future states \(x_{k+1},x_{k+2},\dots,x_{k+H}\) over a horizon \(H\ge 1\), together with uncertainty sets that quantify the discrepancy between the realized stochastic trajectory and the predictor \(\hat f\). The central difficulty is that conformal prediction requires the calibration scores to be representative of the future prediction scores encountered during deployment. We formulate this requirement through two related uncertainty-quantification problems.

\paragraph{Problem 1: One-step rolling uncertainty quantification.}
In the rolling setting, the predictor is repeatedly initialized from the currently observed state. At time \(k\), the one-step nominal prediction is
\begin{equation}
\hat x_{k+1|k} = \hat f(x_k),
\label{eq:auto-display-002}
\end{equation}
and the realized one-step discrepancy is
\begin{equation}
	e(x_k,n_k) := f(x_k)+n_k-\hat f(x_k)
	= x_{k+1}-\hat f(x_k).
	\label{eq:onestep_error}
\end{equation}
After \(x_{k+1}\) is revealed, the predictor is reinitialized from the true state and the procedure repeats. Let \(X \sim \pi\) denote a random state drawn from a deployment state distribution \(\pi\), and let \(X_{+}\sim P(X,\cdot)\) denote the corresponding next state. The goal is to construct a set-valued predictor
\begin{equation}
	\mathcal{C}^{(1)}_{\alpha}(x)\subseteq \mathcal{X}
	\label{eq:onestep_set_problem}
\end{equation}
such that
\begin{equation}
	\mathbb{P}\bigl(X_{+}\in \mathcal{C}^{(1)}_{\alpha}(X)\bigr)
	\ge 1-\alpha,
	\label{eq:onestep_coverage_problem}
\end{equation}
or, more generally, to derive a rigorous bound on the corresponding coverage probability. The set \(\mathcal{C}^{(1)}_{\alpha}(x)\) should be computable from a fixed calibration dataset and should quantify the stochastic one-step discrepancy between \(X_{+}\) and \(\hat f(X)\).

\paragraph{Problem 2: Multi-step uncertainty propagation.}
In the recursive multi-step setting, the learned predictor is applied without reinitialization. Given an initial state \(x_0\), let
\begin{equation}
	\hat x_{h|0}=\hat f^{\,h}(x_0),
	\qquad h=0,\dots,H,
	\label{eq:multistep_nominal_problem}
\end{equation}
denote the nominal prediction, where \(\hat f^{\,0}\) is the identity map and \(\hat f^{\,h}\) is the \(h\)-fold composition of \(\hat f\). Equivalently, \(\hat x_{0|0}=x_0\) and \(\hat x_{h+1|0}=\hat f(\hat x_{h|0})\). The true system evolves according to
\begin{equation}
x_{h+1} = f(x_h)+n_h, \qquad h=0,\dots,H-1.
\label{eq:auto-display-004}
\end{equation}
When the same construction is initialized at a later time \(k\), we use the analogous notation \(\hat x_{k+h|k}=\hat f^{\,h}(x_k)\). The goal is to construct a sequence of uncertainty sets
\begin{equation}
	\mathcal{C}^{(h)}_{\alpha}(x_0),
	\qquad h=1,\dots,H,
	\label{eq:multistep_sets_problem}
\end{equation}
that characterize the discrepancy between the true state \(x_h\), generated by \eqref{eq:dyn_true}, and the nominal prediction \(\hat x_{h|0}\). Unlike the one-step rolling case, the multi-step prediction error is not only a local one-step discrepancy. It also depends on how previous errors are propagated through repeated application of the learned predictor. Hence, multi-step uncertainty sets must combine statistical calibration of local stochastic errors with a dynamical propagation mechanism.

We next address these two objectives through an invariant-measure calibration construction that aligns the calibration law with the dynamical operating regime.

\section{Methodology}	\label{sec:methodology}

In this section, we develop a conformal prediction framework tailored to stochastic dynamical systems. The central idea is to construct the calibration dataset using an invariant measure of the induced Markov process, thereby restoring the statistical structure required for conformal prediction while remaining consistent with the dynamics.

\subsection{Invariant-measure-based calibration and probabilistic structure}

For a probability measure $\mu$ on $(\mathcal{X},\mathcal{B}(\mathcal{X}))$, define the pushforward measure
\begin{equation}
	(\mu P)(A) := \int_{\mathcal{X}}P(x,A)\,\mu(dx), \qquad A\in \mathcal{B}(\mathcal{X}).
\label{eq:auto-equation-002}
\end{equation}

\begin{definition}[Invariant probability measure]
	A probability measure $\mu$ is invariant for the stochastic dynamics if
	\begin{equation}
		\mu P=\mu,
	\label{eq:auto-equation-003}
	\end{equation}
	or equivalently,
	\begin{equation}
		\mu(A)=\int_{\mathcal{X}}P(x,A)\,\mu(dx), \qquad \forall A\in \mathcal{B}(\mathcal{X}).
	\label{eq:auto-equation-004}
	\end{equation}
\end{definition}

The invariance condition implies that if $X_0\sim \mu$, then
\begin{equation}
	X_k\sim \mu, \qquad \forall k\geq 0.
\label{eq:auto-equation-005}
\end{equation}
Thus, $\mu$ defines a stationary distribution of the stochastic dynamical system. This property is central to our approach: it provides a sampling law under which one-step transitions remain statistically consistent over time.

The use of invariant measures is motivated by the fact that many dynamical systems, after transient effects have decayed, operate in recurrent or statistically stationary regimes. In stable stochastic systems, the invariant measure may concentrate around an equilibrium, an attracting set, or a recurrent operating region, depending on the dynamics and the disturbance process. More generally, invariant measures characterize the long-run distribution of states visited by the system, including periodic, quasiperiodic, or chaotic regimes when such stationary descriptions exist \cite{boyarsky2012laws}. By calibrating with respect to $\mu$, the scores are obtained from the stationary behavior of the system rather than from its transient evolution. The resulting conformal bounds are therefore tied to the state distribution encountered during rolling prediction.

\begin{assumption}[Invariant-measure calibration]
	\label{ass:invariant_measure}
	There exists an invariant probability measure $\mu$ for the Markov kernel $P$. The calibration states
	$X_1^{\mathrm{cal}},\ldots,X_n^{\mathrm{cal}}$
	and the test state $X^{\mathrm{test}}$ are sampled independently from $\mu$. Conditional on these states, the corresponding successors are generated according to the same transition kernel $P$.
\end{assumption}

Existence and approximation conditions for invariant measures are recalled in Appendix C.

Under Assumption~\ref{ass:invariant_measure}, we construct the calibration dataset as follows. We sample independent states
\begin{equation}
	X_i^{\mathrm{cal}} \sim \mu, \qquad i=1,\dots,n,
\label{eq:auto-equation-006}
\end{equation}
and for each state generate a successor
\begin{equation}
	X_{i,+}^{\mathrm{cal}} \sim P(X_i^{\mathrm{cal}},\cdot).
\label{eq:auto-equation-007}
\end{equation}

Using the score function introduced in \eqref{eq:score_function}, we define the calibration scores
\begin{equation}
	S_i^{\mathrm{cal}}
	:=
	s_{\hat f}\bigl(X_i^{\mathrm{cal}}, X_{i,+}^{\mathrm{cal}}\bigr),
	\qquad i=1,\dots,n.
\label{eq:auto-equation-008}
\end{equation}

Let $\hat q_{1-\alpha}$ denote the conformal quantile computed from $\{S_i^{\mathrm{cal}}\}_{i=1}^n$. The corresponding one-step prediction set is given by
\begin{equation}
	\mathcal{C}^{(1)}_{\alpha}(x)
	=
	\left\{
	\xi \in \mathcal{X} :
	s_{\hat f}(x,\xi) \le \hat q_{1-\alpha}
	\right\}.
\label{eq:auto-equation-009}
\end{equation}

By construction, the calibration pairs $\bigl(X_i^{\mathrm{cal}}, X_{i,+}^{\mathrm{cal}}\bigr)$ and the test pair $\bigl(X^{\mathrm{test}}, X_{+}^{\mathrm{test}}\bigr)$ are identically distributed under the joint law induced by $\mu$ and $P$. Therefore, provided that the calibration data are not used to train $\hat f$, the calibration scores and the test score
\begin{equation}
s_{\hat f}\bigl(X^{\mathrm{test}}, X_{+}^{\mathrm{test}}\bigr)
\label{eq:auto-display-011}
\end{equation}
are exchangeable, which is the key requirement for conformal validity.

\subsection{One-step rolling prediction}

We now consider the one-step rolling prediction setup. At each time step, the predictor is reinitialized using the true state, and a one-step prediction is performed.
Let the trajectory be generated as
\begin{equation}
x_{k+1} = f(x_k)+n_k, \qquad k \ge 0,
\label{eq:auto-display-012}
\end{equation}
with $x_0 \sim \mu$. Define the test score at time $k$ as
\begin{equation}
	S_k := s_{\hat f}(x_k,x_{k+1}).
\label{eq:auto-equation-010}
\end{equation}

Proofs of all technical results are collected in Appendix E.

\begin{theorem}[One-step rolling conformal guarantee]
	\label{thm:onestep}
	Assume that $\mu$ is invariant for the Markov kernel $P$, and that the calibration states $\{X_i^{\mathrm{cal}}\}_{i=1}^N$ and the initial state $x_0$ are sampled independently from $\mu$, with successors generated according to $P$. Let $\hat q_{1-\alpha}$ be the conformal quantile computed from the calibration scores
	\begin{equation}
	S_i^{\mathrm{cal}}
	=
	s_{\hat f}\bigl(X_i^{\mathrm{cal}},X_{i,+}^{\mathrm{cal}}\bigr),
	\qquad i=1,\dots,N.
	\label{eq:auto-display-013}
	\end{equation}
	Then, for every fixed time step $k \ge 0$,
	\begin{equation}
		\mathbb{P}\bigl( S_k \le \hat q_{1-\alpha} \bigr)
		\ge 1-\alpha .
	\label{eq:auto-equation-011}
	\end{equation}
	Equivalently,
	\begin{equation}
		\mathbb{P}\bigl(x_{k+1}\in \mathcal{C}^{(1)}_{\alpha}(x_k)\bigr)
		\ge 1-\alpha .
	\label{eq:auto-equation-012}
	\end{equation}
\end{theorem}




This result gives a genuine one-step rolling guarantee, since the conformal set remains valid at each fixed time \(k\) when the state distribution is preserved by the invariant measure. If the initial distribution is not invariant, the law of \(x_k\) generally evolves under the Markov dynamics, and the rolling score \(S_k\) may no longer be exchangeable with the calibration scores. Consequently, the conformal guarantee does not automatically extend to later prediction steps. Algorithm 1 in Appendix B summarizes the resulting procedure, where a fixed conformal threshold is computed from invariant-measure calibration samples and applied repeatedly in a rolling one-step prediction setting. 

\subsection{Multi-step prediction}
We now consider multi-step prediction, where the predictor is applied recursively without reinitialization. In contrast to the one-step rolling setting, prediction errors are no longer reset at each step, but instead propagate through the model. As a result, both statistical uncertainty, captured by the conformal scores, and dynamical amplification, captured by the recursive application of $\hat f$, must be taken into account. In this subsection we use the shorthand $\hat x_h := \hat x_{h|0}$ for the nominal prediction introduced in \eqref{eq:multistep_nominal_problem}.

Starting from the common initial condition $x_0=\hat x_0$, the true and predicted trajectories evolve as
\begin{equation}
x_{h+1} = f(x_h)+n_h,
\qquad
\hat x_{h+1} = \hat f(\hat x_h),
\label{eq}
\end{equation}
respectively. The prediction error at horizon $h$ is then defined by
\begin{equation}
e_h := x_h - \hat x_h.
\label{eq}
\end{equation}

\begin{assumption}
	The predictor $\hat f$ is $L$-Lipschitz, i.e.,
	\begin{equation}
		\|\hat f(x) - \hat f(\xi)\| \le L \|x-\xi\|,
		\qquad \forall x,\xi \in \mathcal{X}.
	\label{eq:auto-equation-013}
	\end{equation}
\end{assumption}

\begin{remark}
	The Lipschitz condition is imposed only on the predictor \(\hat f\), not on the true dynamics \(f\). This is sufficient because the multi-step error bound depends on the repeated application of \(\hat f\), so error amplification is governed by the sensitivity of the learned predictor. In practice, this assumption can be enforced or promoted during model construction. For neural-network predictors, one may use spectral normalization, spectral-norm regularization, Parseval-type constraints, or Lipschitz-constrained architectures \cite{miyato2018spectral,yoshida2017spectral,cisse2017parseval,anil2019sorting}. The resulting Lipschitz constant can also be estimated or upper-bounded using optimization-based certification methods \cite{fazlyab2019efficient}. Thus, even when the true dynamics are complex, the predictor used for uncertainty propagation can be designed, regularized, or certified within a Lipschitz class.
\end{remark}

To relate the multi-step error to the one-step conformal scores, define
\begin{equation}
S_h := s_{\hat f}(x_h,x_{h+1}), \qquad h \ge 0.
\label{eq:auto-display-021}
\end{equation}

We first establish how the error evolves from one step to the next.

\begin{lemma}[Error recursion]
	\label{lem:error_recursion}
	For all $h \ge 0$,
	\begin{equation}
		\|e_{h+1}\| \le L\|e_h\| + S_h.
	\label{eq:auto-equation-014}
	\end{equation}
\end{lemma}

The previous lemma shows that the error consists of two components at each step: a propagated error term and a newly introduced one-step discrepancy. We now unroll this recursion to obtain an explicit expression for the error after $k$ steps.

\begin{lemma}[Multi-step error expansion]
	\label{lem:kstep_bound}
	For any integer $k \ge 1$,
	\begin{equation}
		\|e_k\| \le \sum_{i=0}^{k-1} L^{k-1-i} S_i.
	\label{eq:auto-equation-015}
	\end{equation}
\end{lemma}

\begin{remark}[Contractive predictors]
	If the learned predictor is contractive, then \(L<1\) and the propagation factor in \eqref{eq:gamma_k} satisfies \(\Gamma_k(L)=\sum_{i=0}^{k-1}L^i\to 1/(1-L)\) as \(k\to\infty\). Thus, the radii \(\hat q_{1-\alpha}\Gamma_k(L)\) remain uniformly bounded with the horizon, instead of growing without bound. This is desirable in recursive prediction because the uncertainty tube does not expand solely due to repeated application of the predictor.	
	
\end{remark}

We now translate this pathwise bound into a probabilistic statement for an arbitrary horizon. The goal is to control the event that all one-step scores along the trajectory remain below the conformal threshold.

\begin{lemma}[$k$-step probability bound]
	\label{lem:kstep_prob}
	Let
	\begin{equation}
	A_i := \{S_i \le \hat q_{1-\alpha}\}, \qquad i = 0,\dots,k-1.
	\label{eq:auto-display-035}
	\end{equation}
	Then
	\begin{equation}
		\mathbb{P}\left(\bigcap_{i=0}^{k-1} A_i\right)
		\ge
		1-\min\{1,k\alpha\}.
	\label{eq:auto-equation-017}
	\end{equation}
\end{lemma}

We now combine the pathwise and probabilistic results.

\begin{theorem}[Multi-step conformal guarantee]
	\label{thm:multistep}
	Under the invariant-measure calibration construction of Theorem~\ref{thm:onestep}, suppose that $\hat f$ is $L$-Lipschitz. Then, for any $k \ge 1$,
	\begin{equation}
		\mathbb{P}\left(
		\bigcap_{j=1}^{k}
		\left\{
		\|x_j - \hat x_j\| \le \hat q_{1-\alpha}\,\Gamma_j(L)
		\right\}
		\right)
		\ge
		1-\min\{1,k\alpha\},
		\label{eq:multistep_thm}
	\end{equation}
	where
	\begin{equation}
		\Gamma_k(L):=
		\begin{cases}
			\dfrac{L^k-1}{L-1}, & L\neq 1,\\[1em]
			k, & L=1.
		\end{cases}
		\label{eq:gamma_k}
	\end{equation}
	In particular, the terminal-horizon coverage also satisfies
	\begin{equation}
		\mathbb{P}\left(
		\|x_k - \hat x_k\| \le \hat q_{1-\alpha}\,\Gamma_k(L)
		\right)
		\ge
		1-\min\{1,k\alpha\}.
		\label{eq:multistep_terminal_thm}
	\end{equation}
\end{theorem}

For clarity, Algorithm 2 in Appendix B summarizes the overall procedure corresponding to the multi-step prediction setup. The procedure uses the same calibration object together with the Lipschitz-based propagation bound to generate multi-step uncertainty sets along the prediction horizon.

\section{Simulation and Results}\label{sec:simulation}

In this section, we evaluate imCP on a nonlinear rotational stochastic benchmark system with localized angular forcing. The full system equations, Table 2 of numerical parameters, invariant-measure approximation, predictor specification, and evaluation formulas are provided in Appendix D. We first evaluate the rolling one-step setting, where the predictor is reinitialized at every time step using the true observed state and tested on independent trajectories initialized from the invariant measure.
\begin{figure}[t]
	\centering
	\includegraphics[scale=.4,trim={0cm 0cm .5cm .0cm},clip]{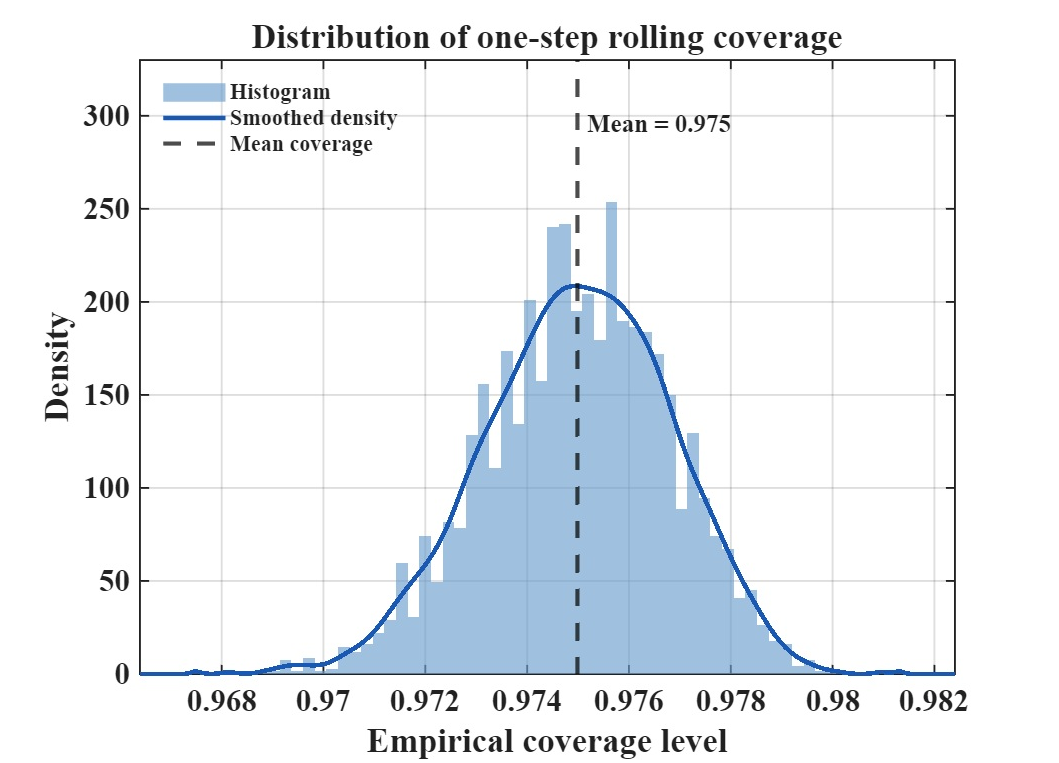}
	\caption{Empirical one-step rolling coverage of imCP for the rotational stochastic benchmark system. Coverage concentrates near the target level.}
	\label{fig:imperical_coverage}
\end{figure}
The empirical one-step coverage is computed across multiple splits. In imCP, calibration and testing are both based on the invariant distribution. Since this distribution is preserved by the dynamics, the calibration scores and the rolling deployment scores have the same distribution and are exchangeable. Consequently, the empirical coverage concentrates near the target value $1-\alpha$, with $\alpha = 0.025$. Fig.~\ref{fig:imperical_coverage} shows the resulting distribution of one-step rolling coverage values.

To study the effect of the calibration distribution, we compare four calibration schemes while keeping the deployment setting fixed. Fig.~\ref{fig:one_step_time_coverage} shows the resulting rolling coverage over time for the different calibration laws. When the calibration samples are drawn from a distribution that does not match the invariant distribution used during deployment, the calibration residuals may no longer represent the errors observed during rolling prediction. Appendix D gives the detailed construction of each calibration law and the empirical rolling-coverage estimator, and Table 3 in Appendix D reports the corresponding numerical comparison, including coverage, split-to-split standard deviation, average conformal radius, and the gap to the target coverage level for each calibration scheme.

\begin{figure}[t]
	\centering
	\includegraphics[scale=.43]{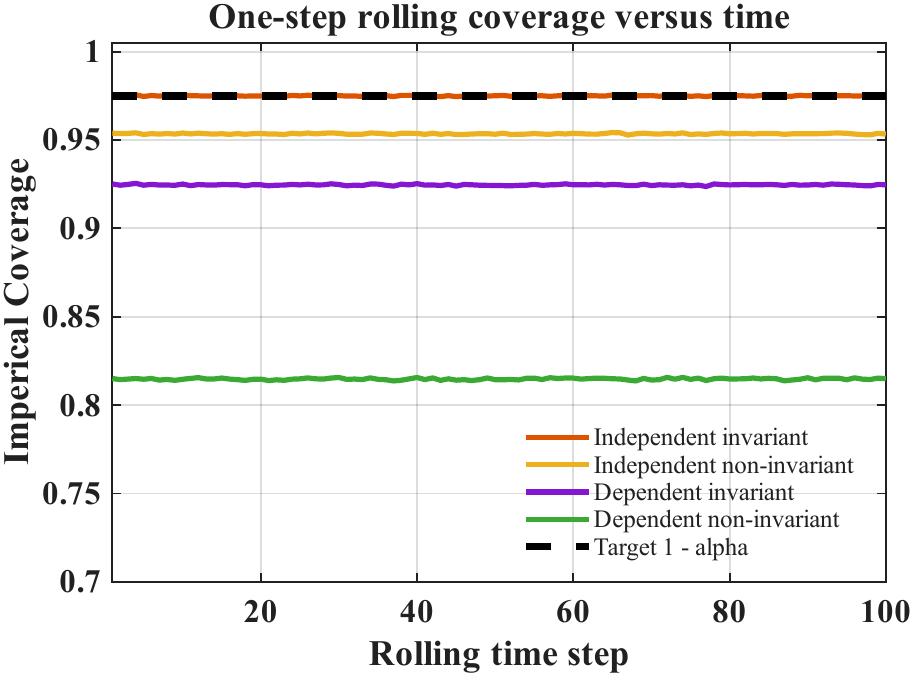}
	\caption{One-step rolling coverage for different calibration baselines. Invariant calibration preserves coverage, while misaligned calibration degrades over time.}
	\label{fig:one_step_time_coverage}
\end{figure}

Having established the effect of calibration alignment in the rolling one-step setting, we next turn to local recursive multi-step prediction over $H_{\rm loc}=5$ steps. We compare imCP with one-step conformal prediction~\cite{vovk2012conditional}, fixed-horizon conformal prediction~\cite{lindemann2023safe}, trajectory-max conformal prediction~\cite{cleaveland2024conformal}, PID conformal prediction~\cite{angelopoulos2023pid}, and an EnbPI-style baseline~\cite{xu2021enbpi} under the same invariant deployment regime; the detailed multi-step calibration setup is given in Appendix D.

\begin{figure}[t]
	\centering
	\includegraphics[scale=0.38,trim={3cm 1cm 5cm 3.75cm},clip]{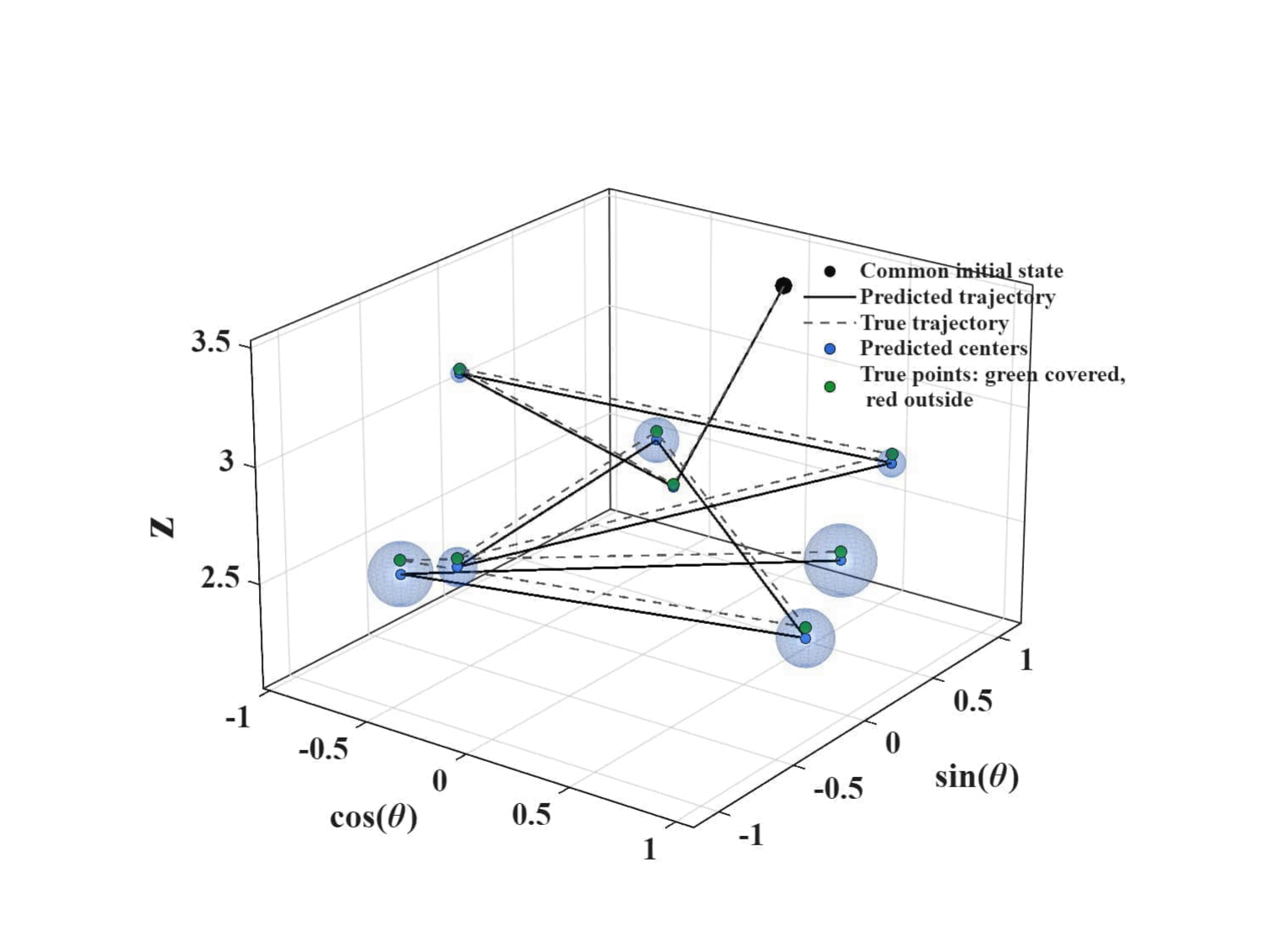}
	\caption{Recursive multi-step prediction with imCP propagated uncertainty balls for one split. The radii grow with the horizon to account for accumulated prediction error.}
	\label{fig:multi_step_uncertainty}
\end{figure}

For each split, we evaluate whether the true trajectory remains within each candidate bound, as illustrated in Fig.~\ref{fig:multi_step_uncertainty}. We report two coverage notions in order to distinguish pointwise-in-horizon reliability from simultaneous reliability over the local rollout. The horizon-wise coverage at step $h$ is the fraction of test trajectories satisfying $\|x_h-\hat x_h\|\le r_h$. The local pathwise coverage is the fraction satisfying all local coverage events simultaneously,
\begin{equation}
\resizebox{\columnwidth}{!}{$
{\color{black}
\begin{aligned}
\widehat{\operatorname{Cov}}_{\rm path}
&=
\frac{1}{N_{\rm split}N_{\rm test}}
\sum_{b=1}^{N_{\rm split}}\sum_{i=1}^{N_{\rm test}}
\mathbf{1}\!\left\{
\begin{array}{l}
\|x_{i,h}^{(b)}-\hat x_{i,h}^{(b)}\|\le r_{i,h}^{(b)},\\
\qquad h=1,\dots,H_{\rm loc}
\end{array}
\right\}.
\end{aligned}
}
$}
\label{eq:local_pathwise_coverage}
\end{equation}
Thus, a method can have high marginal coverage at each individual horizon but lower pathwise coverage, because pathwise coverage requires no miss anywhere along the local trajectory. Misses that occur at different horizons on different trajectories accumulate through this intersection event.

Table 4 in Appendix D summarizes the local five-step comparison across empirical coverage, uncertainty size, and calibration cost. The table reports final-horizon coverage, pathwise coverage, average radius, average volume, and the number of calibrated quantiles. As detailed in Appendix B, imCP provides a theoretical pathwise guarantee for the propagated tube in \eqref{eq:local_pathwise_coverage}; trajectory-max conformal prediction provides a separate trajectory-level guarantee, while the other methods report empirical pathwise coverage only.

\begin{figure}[t]
	\centering
	\includegraphics[scale=.43]{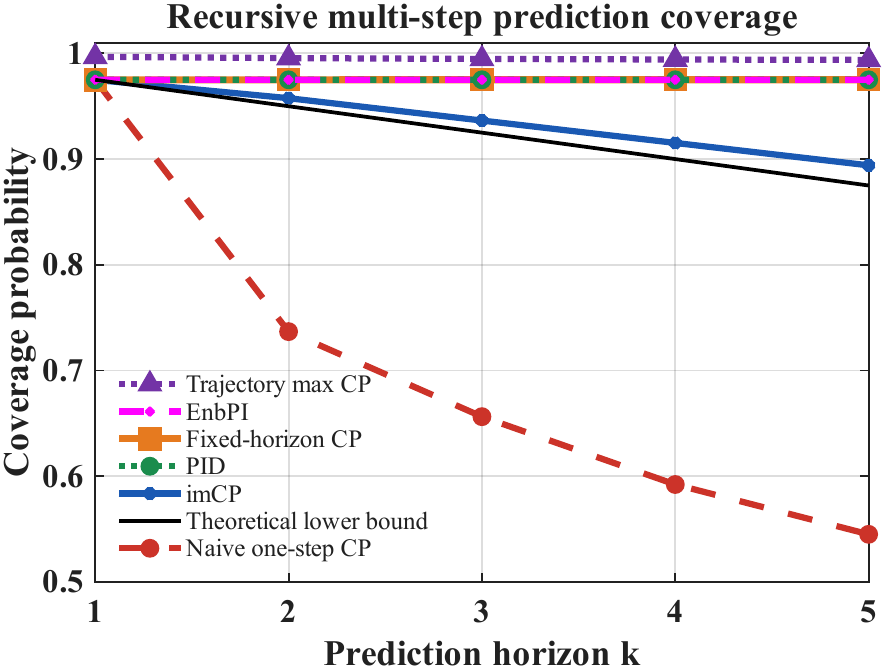}
	\caption{Empirical recursive multi-step coverage for imCP, comparison methods, and the bound $1-k\alpha$. The imCP method stays above its pathwise lower bound across the horizon.}
	\label{fig:multi_step_performance}
\end{figure}

The coverage results in Fig.~\ref{fig:multi_step_performance} confirm the
expected qualitative behavior. The naive one-step conformal radius rapidly loses
coverage as the horizon increases because it does not account for recursive
error growth. Fixed-horizon conformal prediction maintains high horizon-wise
coverage, but this does not by itself give the pathwise guarantee in
\eqref{eq:multistep_thm}: each horizon is calibrated separately, so the
trajectory-level coverage can be lower than the stepwise coverages shown in the
plot. Trajectory-max conformal prediction does target a trajectory-level event,
but it does so by calibrating the maximum error over the full local horizon and
therefore becomes conservative, especially at early prediction steps. The PID
and EnbPI baselines update their horizon-wise radii from previously observed
test residuals or residual batches, and their curves therefore reflect online
adaptation within the ordered test sequence rather than a fully pre-deployment
tube.

The imCP method differs from these baselines in that the same one-step
conformal event controls the whole propagated tube through the Lipschitz factor.
Consequently, Theorem~\ref{thm:multistep} gives a lower bound on the
trajectory-level coverage of the imCP tube, not only on the terminal
horizon. This theoretical lower bound is specific to the imCP construction;
it does not certify the fixed-horizon, PID, or EnbPI curves, while the
trajectory-max baseline has a separate trajectory-level calibration mechanism.
Empirically, imCP remains above the theoretical bound and has
pathwise coverage comparable to the fixed-horizon baseline, while requiring only
one conformal quantile and producing substantially smaller radii than the
trajectory-max construction.

The coverage comparison must be read together with the size of the resulting
sets. Fig.~\ref{fig:local_radius_volume} provides this complementary view by
showing the average radius produced by each method. The naive one-step baseline
has the smallest radius because it keeps the one-step threshold fixed, but this
small radius is not sufficient to maintain coverage in
Fig.~\ref{fig:multi_step_performance}. The fixed-horizon and trajectory-max
methods achieve high coverage, but their radii are substantially larger. The
trajectory-max method is especially conservative because it uses a radius
calibrated for the worst error over the whole local horizon and applies it at
every step. The online PID and EnbPI radii vary with the residuals observed
earlier in the test sequence, so their average radii summarize adaptive behavior
under the same trajectory ordering used for the coverage plot. imCP lies
between the undercovered naive baseline and the more conservative
trajectory-level or horizon-specific constructions. Its radius increases with
the prediction horizon, reflecting recursive error propagation. Taken together,
Table 4 in Appendix D and
Figs.~\ref{fig:multi_step_performance}--\ref{fig:local_radius_volume} show that,
in this experiment, imCP provides an explicit trajectory-level
lower bound while avoiding the largest radii produced by the more conservative
calibration rules.

\begin{figure}[t]
	\centering
	\includegraphics[scale=.43]{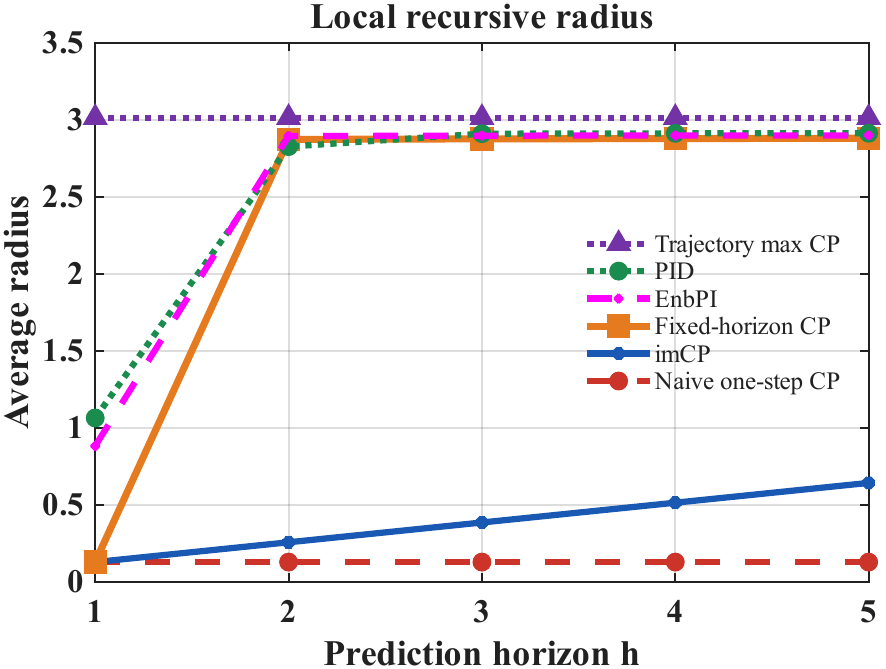}
	\caption{Average uncertainty radius across horizons for the local five-step comparison. The imCP method avoids the largest radii while accounting for recursive error growth.}
	\label{fig:local_radius_volume}
\end{figure}

Overall, the simulations support imCP in both one-step rolling and recursive multi-step settings. In the one-step case, invariant-measure calibration aligns calibration and deployment scores, while non-invariant or temporally dependent calibration can break the exchangeability structure required by split conformal prediction and degrade coverage. In the multi-step case, the Lipschitz propagation bound extends a single one-step conformal calibration to recursive prediction, producing meaningful horizon-dependent uncertainty sets without separate calibration at each horizon.

\section{Conclusion and Future Work} \label{sec:conclusion}


This paper developed imCP for stochastic dynamical systems. By sampling calibration states independently from an invariant measure of the induced Markov process, imCP aligns the calibration score distribution with the rolling deployment regime and recovers the exchangeability structure needed for split conformal prediction. We established a finite-sample one-step rolling coverage guarantee and extended it to recursive multi-step prediction by combining the calibrated one-step error radius with a Lipschitz-based propagation bound and Boole's inequality. The simulations support the theoretical findings. Invariant calibration preserves coverage in rolling prediction, while non-invariant or temporally dependent calibration schemes can lead to coverage degradation. In the multi-step setting, the imCP propagated radius provides meaningful horizon-dependent uncertainty sets using only a one-step calibration object, avoids separate fixed-horizon calibration, and remains less conservative than trajectory-level worst-case calibration.


Several directions remain open. A first direction is to relax the exact invariant-measure sampling requirement by quantifying how approximation error, finite trajectory sampling, or transient deviations from invariance affect coverage. A second direction is to sharpen the multi-step bounds by replacing global Lipschitz constants with local, data-dependent, or invariant-set-dependent propagation estimates, which may reduce conservatism. Extending imCP to controlled systems is also important, since changes in feedback policy can alter the invariant operating distribution and invalidate a calibration law constructed under a previous policy. Finally, imCP can be further developed for uncertainty-aware observer design, fault detection and isolation, self-triggered control, and safety verification, where pre-deployment uncertainty tubes are needed before future residuals are observed.
	

\clearpage
\def\incombinedpaper{}
\ifdefined\incombinedpaper
\else
\documentclass[letterpaper]{article}
\usepackage[submission]{aaai2027}
\usepackage[hyphens]{url}
\usepackage{graphicx}
\urlstyle{rm}
\def\UrlFont{\rm}
\usepackage{natbib}
\usepackage{caption}
\frenchspacing
\usepackage{amsmath}
\usepackage{amssymb}
\usepackage{amsthm}
\usepackage{amsfonts}
\usepackage{stackengine}
\usepackage{booktabs}
\usepackage{array}
\usepackage{multirow}
\usepackage{pifont}
\usepackage{pdflscape}
\usepackage{color}
\usepackage{float}
\usepackage{rotating}
\usepackage[linesnumbered,ruled]{algorithm2e}

\pdfinfo{
/TemplateVersion (2027.1)
}

\setcounter{secnumdepth}{0}

\newtheorem{theorem}{Theorem}
\newtheorem{assumption}{Assumption}
\newtheorem{remark}{Remark}
\newtheorem{proposition}{Proposition}
\newtheorem{lemma}{Lemma}
\newtheorem{definition}{Definition}
\newcommand{\revblue}[1]{#1}
\definecolor{forestgreen}{rgb}{0.13, 0.55, 0.13}
\newcommand{\greencheck}{\textcolor{forestgreen}{\checkmark}}
\newcommand{\redcross}{\textcolor{red}{\(\times\)}}

\title{Technical Appendix: Uncertainty Quantification via Invariant-Measure Conformal Prediction}
\author{
    Mohammadhossein~Bakhtiaridoust,
    Dominik~Baumann,
    and Shankar~Deka
}
\affiliations{
    Anonymous AAAI-27 submission
}

\begin{document}
\maketitle
\fi
\appendix
\section{Appendix A: Split Conformal Prediction Review}

Conformal prediction is a general framework for quantifying uncertainty in the predictions made by arbitrary prediction models \cite{shafer2008tutorial,angelopoulos2022gentle,angelopoulos2025foundations}. Among the different conformal constructions, split conformal prediction, also known as inductive conformal prediction, is particularly attractive in practice because it achieves computational efficiency by separating model fitting from calibration \cite{vovk2012conditional}.

Consider a supervised learning setting with data points
\begin{equation}
	d_i := (u_i, v_i), \qquad i = 1,\dots,n,
	\label{eq:data_points}
\end{equation}
where $u_i \in \mathcal{U}$ denotes the generic input (or feature vector) and $v_i \in \mathcal{V} \subseteq \mathbb{R}^{d_v}$ denotes the corresponding generic response.

The dataset is partitioned into two disjoint subsets. The first subset consists of $n_{\mathrm{tr}}$ samples used for training, and the remaining $n_{\mathrm{cal}} := n - n_{\mathrm{tr}}$ samples are used for calibration. Specifically, the training set is given by $\{(u_i,v_i)\}_{i=1}^{n_{\mathrm{tr}}}$ and the calibration set is given by $\{(u_j,v_j)\}_{j=n_{\mathrm{tr}}+1}^{n}$.
A predictor
\begin{equation}
	\hat g : \mathcal{U} \rightarrow \mathcal{V},
	\label{eq:predictor}
\end{equation}
is first constructed using only the training set. The role of the calibration set is then to quantify the discrepancy between the predictor $\hat g$ and the observed responses.

To this end, we define a nonconformity (or score) function
\begin{equation}
	s_{\hat{g}} : \mathcal{U} \times \mathcal{V} \rightarrow \mathbb{R}_{\ge 0},
	\label{eq:generic_score_function}
\end{equation}
where $s_{\hat{g}}(u,v)$ measures how well the response $v$ agrees with the prediction $\hat g(u)$. For each calibration point $(u_j,v_j)$, we compute the score
\begin{equation}
	S_j := s_{\hat{g}}(u_j,v_j), 
	\qquad j = n_{\mathrm{tr}}+1,\dots,n.
	\label{eq:calibration_scores}
\end{equation}

In the scalar regression setting, a common choice of score is the absolute residual
\begin{equation}
	S_j = |v_j - \hat g(u_j)|,
	\label{eq:absolute_residual}
\end{equation}
while in more general settings one may consider weighted or asymmetric scores, or quantile-based scores as in conformalized quantile regression \cite{romano2019cqr,angelopoulos2022gentle}.

Given an $\alpha \in (0,1)$, let $1-\alpha$ denote the desired coverage level. The conformal method constructs a threshold based on the empirical distribution of the calibration scores. Specifically, let $\hat q_{1-\alpha}$ be the $\left\lceil (n_{\mathrm{cal}}+1)(1-\alpha)\right\rceil$-th smallest value in $\{S_j\}_{j=n_{\mathrm{tr}}+1}^{n}$.
This definition includes a finite-sample correction that is essential for achieving exact marginal coverage guarantees.

This quantity is an \emph{order statistic} of the calibration scores \cite{shafer2008tutorial}. The use of order statistics is central to conformal prediction, as it ensures that the resulting prediction set achieves finite-sample coverage.

Given a new generic input $u_{n+1}$, the conformal prediction set is defined as
\begin{equation}
	\mathcal{C}_{\alpha}(u_{n+1})
	=
	\left\{v \in \mathcal{V} :
	s_{\hat{g}}(u_{n+1},v) \le \hat q_{1-\alpha}
	\right\}.
	\label{eq:prediction_set}
\end{equation}

In the case of the absolute residual score \eqref{eq:absolute_residual}, this reduces to the interval
\begin{equation}
	\mathcal{C}_{\alpha}(u_{n+1})
	=
	\bigl[
	\hat g(u_{n+1}) - \hat q_{1-\alpha},
	\;
	\hat g(u_{n+1}) + \hat q_{1-\alpha}
	\bigr].
	\label{eq:prediction_interval}
\end{equation}

Under the assumption that the calibration points $\{(u_j,v_j)\}$ and the test point $(u_{n+1},v_{n+1})$ are exchangeable, the conformal prediction set satisfies the finite-sample coverage guarantee
\begin{equation}
	\mathbb{P}\bigl(
	v_{n+1} \in \mathcal{C}_{\alpha}(u_{n+1})
	\bigr)
	\ge 1 - \alpha,
	\label{eq:coverage}
\end{equation}
where the probability is taken over both the calibration data and the test point, yielding a marginal coverage guarantee.

The key reason behind this guarantee is that, under exchangeability, the rank of the test score
$
s_{\hat{g}}(u_{n+1},v_{n+1})
$
among the $n_{\mathrm{cal}}+1$ values
$
\{S_j\}_{j=n_{\mathrm{tr}}+1}^{n}
\cup
\{s_{\hat{g}}(u_{n+1},v_{n+1})\}
$
is uniformly distributed. Therefore, the probability that the test score exceeds the empirical $(1-\alpha)$ quantile is controlled by $\alpha$, leading to the coverage guarantee \eqref{eq:coverage}.

This finite-sample guarantee is a central feature of split conformal prediction: once the predictor is fixed, the coverage statement depends on the exchangeability of the calibration and test scores, and not on the correctness of the predictive model itself \cite{shafer2008tutorial,angelopoulos2022gentle}. Therefore, the calibration procedure must be chosen so that the empirical score distribution used to define the conformal quantile is representative of the score encountered at deployment.

\section{Appendix B: Method Comparison and Algorithms}

In this section, we provide pseudocode for the imCP algorithms and a detailed comparison of imCP with existing conformal prediction methods for stochastic dynamical systems.

\begin{table*}[tbp]
	\centering
	\scriptsize
	\setlength{\tabcolsep}{2.0pt}
	\renewcommand{\arraystretch}{1.08}
	\caption{Comparison of conformal prediction methods for stochastic dynamical systems}
	\label{tab:comparison}
	\resizebox{\textwidth}{!}
	{%
	\begin{tabular}{lcccccc}
		\hline
		\textbf{Method} 
		& \textbf{Rolling} 
		& \textbf{Recursive} 
		& \textbf{Pre-deployment} 
		& \textbf{Guarantee type} 
		& \textbf{Pathwise} 
		& \textbf{Main limitation} \\
		& \textbf{prediction} 
		& \textbf{multi-step} 
		& \textbf{bound} 
		& 
		& \textbf{guarantee} 
		& \\
		\hline
		Full conformal\\~\cite{shafer2008tutorial,angelopoulos2025foundations}
		& \redcross 
		& \redcross 
		& \redcross 
		& Finite-sample marginal 
		& \redcross 
		& Computationally costly \\
		Split CP~\cite{vovk2012conditional}
		& \redcross 
		& \redcross 
		& \greencheck 
		& Finite-sample marginal 
		& \redcross 
		& No mechanism for dynamical distribution alignment \\
		CQR~\cite{romano2019cqr}
		& \redcross 
		& \redcross 
		& \greencheck 
		& Finite-sample marginal 
		& \redcross 
		& Calibration-distribution dependent \\
		Weighted CP~\cite{tibshirani2019covariate,barber2023beyond}
		& \redcross 
		& \redcross 
		& \greencheck 
		& Weighted marginal 
		& \redcross 
		& Requires calibration-to-deployment density ratio \\
		ACI / online CP\\~\cite{gibbs2021aci,gibbs2024online,zaffran2022adaptive,angelopoulos2023pid}
		& \greencheck 
		& \redcross 
		& \redcross 
		& Long-run coverage 
		& \redcross 
		& Requires online feedback \\
		CF-RNN~\cite{stankeviciute2021conformal}
		& \redcross 
		& \greencheck 
		& \greencheck 
		& Horizon-wise coverage 
		& \redcross 
		& Requires sequence assumptions \\
		EnbPI~\cite{xu2021enbpi}
		& \greencheck 
		& \greencheck 
		& \redcross 
		& Approximate marginal 
		& \redcross 
		& Relies on mixing assumptions \\
		Fixed-horizon CP~\cite{lindemann2023safe}
		& \redcross 
		& \greencheck 
		& \greencheck 
		& Fixed-horizon marginal 
		& \redcross 
		& Horizon specific \\
		Trajectory-level CP~\cite{cleaveland2024conformal}
		& \redcross 
		& \greencheck 
		& \greencheck 
		& Trajectory-level 
		& \greencheck 
		& Conservative worst-case radius \\
		\hline
		\textbf{imCP}
		& \greencheck 
		& \greencheck 
		& \greencheck 
		& Invariant-regime marginal 
		& \greencheck 
		& Requires invariant sampling \\
		\hline
	\end{tabular}
	}
\end{table*}
\makeatletter
\setlength{\@fptop}{0pt}
\setlength{\@fpsep}{8pt plus 1fil}
\setlength{\@fpbot}{0pt plus 1fil}
\makeatother

\begin{algorithm}[htbp]
	\DontPrintSemicolon
	\footnotesize
	\caption{imCP calibration for one-step rolling prediction}\label{alg:onesteprolling_app}
	\SetKwInOut{Input}{Input}
	\SetKwInOut{Output}{Output}
	\Input{Transition kernel $P$; predictor $\hat f$; invariant measure $\mu$; calibration size $N$; miscoverage level $\alpha \in (0,1)$}
	\Output{One-step rolling prediction sets $\{\mathcal{C}^{(1)}_{\alpha}(x_k)\}_{k\ge 0}$}
	
	Sample calibration states $X_1^{\mathrm{cal}},\dots,X_N^{\mathrm{cal}} \overset{\mathrm{i.i.d.}}{\sim} \mu$\;
	
	\For{$i \gets 1$ \KwTo $N$}{
		Observe or simulate successor $X_{i,+}^{\mathrm{cal}}\sim P(X_i^{\mathrm{cal}},\cdot)$\;
		Compute calibration score $S_i^{\mathrm{cal}} \gets s_{\hat f}(X_i^{\mathrm{cal}},X_{i,+}^{\mathrm{cal}})$\;
	}
	
	Compute the conformal quantile $\hat q_{1-\alpha}$ from $\{S_i^{\mathrm{cal}}\}_{i=1}^N$\;
	
	Initialize the rolling test trajectory with $x_0\sim\mu$\;
	
	\ForEach{time step $k=0,1,2,\dots$}{
		Observe the true current state $x_k$\;
		Compute one-step prediction $\hat x_{k+1|k} \gets \hat f(x_k)$\;
		Form the one-step prediction set $\mathcal{C}^{(1)}_{\alpha}(x_k)$\;
		Output $\hat x_{k+1|k}$ and $\mathcal{C}^{(1)}_{\alpha}(x_k)$\;
	}
\end{algorithm}

\begin{algorithm}[hbp]
	\DontPrintSemicolon
	\caption{imCP for multi-step forecasting}\label{alg:multistep_app}
	\SetKwInOut{Input}{Input}
	\SetKwInOut{Output}{Output}
	\Input{Transition kernel $P$; predictor $\hat f$; invariant measure $\mu$; calibration size $N$; miscoverage level $\alpha \in (0,1)$; prediction horizon $H$; Lipschitz constant $L$ of $\hat f$}
	\Output{Multi-step predicted trajectory $\{\hat x_{k+h|k}\}_{h=1}^H$ and uncertainty sets $\{\mathcal{C}^{(h)}_{\alpha}(x_k)\}_{h=1}^H$}
	
	Sample calibration states $X_1^{\mathrm{cal}},\dots,X_N^{\mathrm{cal}} \overset{\mathrm{i.i.d.}}{\sim} \mu$\;
	
	\For{$i \gets 1$ \KwTo $N$}{
		Observe or simulate successor $X_{i,+}^{\mathrm{cal}}\sim P(X_i^{\mathrm{cal}},\cdot)$\;
		Compute calibration score $S_i^{\mathrm{cal}} \gets s_{\hat f}(X_i^{\mathrm{cal}},X_{i,+}^{\mathrm{cal}})$\;
	}
	
	Compute the conformal quantile $\hat q_{1-\alpha}$ from $\{S_i^{\mathrm{cal}}\}_{i=1}^N$\;
	
	\ForEach{prediction time $k$}{
		Observe the true initial state $x_k$\;
		Initialize $\hat x_{k|k} \gets x_k$\;
		
		\For{$h \gets 0$ \KwTo $H-1$}{
			Propagate the predictor $\hat x_{k+h+1|k} \gets \hat f(\hat x_{k+h|k})$\;
			Set the uncertainty radius $r_{h+1}$ using the multi-step bound in the main paper\;
			Form $\mathcal{C}^{(h+1)}_{\alpha}(x_k)$ as the ball centered at
			$\hat x_{k+h+1|k}$ with radius $r_{h+1}$\;
		}
		
		Output $\{\hat x_{k+h|k}\}_{h=1}^{H}$ and $\{\mathcal{C}^{(h)}_{\alpha}(x_k)\}_{h=1}^{H}$\;
	}
\end{algorithm}

\section{Appendix C: Invariant-Measure Existence and Approximation}

The existence of the invariant law used by imCP is not restrictive under standard compactness and continuity conditions.

\begin{lemma}[Krylov--Bogolyubov \cite{lasota2013chaos}]
	\label{lem:invariant_measure_existence_app}
	Suppose that $\mathcal{X}$ is a compact metric space and that
	the Markov kernel $P$ is Feller. Then there exists at least one
	$P$-invariant probability measure $\mu\in\mathcal{P}(\mathcal{X})$.
\end{lemma}

\begin{remark}[Obtaining or approximating the invariant measure]
	The invariant measure is available explicitly in several standard settings. For finite-state Markov chains, invariant distributions are obtained from the stationary balance equations of the transition matrix \cite{meyn2009markov}. For stable affine stochastic systems, the invariant law is Gaussian, with mean and covariance determined by the steady-state equations associated with the dynamics \cite{meyn2009markov}. Certain deterministic or stochastic maps also admit closed-form invariant densities through the Frobenius--Perron equation; classical examples include expanding maps such as the doubling map and Chebyshev-type maps \cite{lasota1994chaos}.
	When a closed-form expression is not available, the invariant measure can often be approximated numerically. Ulam's method approximates the Frobenius--Perron operator by a finite Markov matrix induced by a partition of the state space, whose stationary distribution provides a finite-dimensional approximation of the invariant measure \cite{li1976ulam}. Convergence of such approximations is known for important classes of maps, including multidimensional extensions \cite{li1976ulam,ding1996finite}. Alternatively, for ergodic Markov processes, long-run empirical occupation measures converge asymptotically to the invariant law \cite{meyn2009markov}. Thus, the calibration law required by invariant-measure calibration can be obtained exactly in analytically tractable cases and approximated in more general settings.
\end{remark}

\section{Appendix D: Simulation System and Evaluation Details}

The numerical experiments use a nonlinear rotational stochastic benchmark system with localized angular forcing. We consider the discrete-time system
\begin{equation}
	 \theta_{k+1} = \theta_k + \omega \;\; \mathrm{mod}\; 2\pi, \qquad z_{k+1} = \rho z_k + \varepsilon_k + d(\theta_k),
	 \label{eq:phase_system_app}
\end{equation}
The prediction state is the embedded coordinate \(x_k \in \mathbb{R}^3\), defined from the angular coordinate and scalar auxiliary state by
\begin{equation}
x_k =
\begin{bmatrix}
\cos \theta_k \\
\sin \theta_k \\
z_k
\end{bmatrix}.
\label{eq:embedded_state_app}
\end{equation}
Here, $\theta_k \in [0,2\pi)$ is an angular coordinate, $z_k \in \mathbb{R}$ is an auxiliary state, and $\omega$ is chosen irrational, which ensures ergodicity of the rotation on the circle. The parameter $\rho \in (0,1)$ defines a stable contraction, $\varepsilon_k$ is a small bounded noise term, and $d(\theta_k)$ is a localized angular forcing term given by
\begin{equation}
d(\theta_k) =
\begin{cases}
A, & \theta_k \in [0, 2\pi \beta],\\
0, & \text{otherwise}.
\end{cases}
\label{eq:phase_disturbance_app}
\end{equation}

The numerical parameter values used in the experiments are summarized in Table~\ref{tab:simulation_parameters}.

\begin{table}[H]
	\centering
	\small
	\setlength{\tabcolsep}{2.4pt}
	\renewcommand{\arraystretch}{1.05}
	\caption{\textcolor{black}{Simulation parameters for the rotational stochastic benchmark system.}}
	\label{tab:simulation_parameters}
	\textcolor{black}{%
	\begin{tabular}{lll}
		\hline
		\textbf{Parameter} & \textbf{Value} & \textbf{Description} \\
		\hline
		\(\rho\) & \(0.95\) & Stable contraction factor \\
		\(\varepsilon_k\) & \(\mathrm{Unif}[-0.02,0.02]\) & Process noise \\
		\(\omega\) & \(2\pi(\sqrt{5}-1)/2\) & Irrational rotation increment \\
		\(A\) & \(3.0\) & Disturbance amplitude \\
		\(\beta\) & \(0.022\) & Active phase fraction \\
		\hline
	\end{tabular}}
\end{table}

\begin{table}[b]
	\centering
	\small
	\setlength{\tabcolsep}{3.4pt}
	\renewcommand{\arraystretch}{1.05}
	\caption{One-step rolling coverage under different calibration schemes.}
	\label{tab:one_step_calibration_laws}
	\begin{tabular}{@{}
	>{\raggedright\arraybackslash}p{3.0cm}
	cccc@{}}
	\hline
	\textbf{Calib. scheme} & \textbf{Cov.} & \textbf{Std.} & \textbf{Radius} & \textbf{Gap} \\
	\hline
	Independent invariant & 0.975 & 0.002 & 0.125 & 0.000 \\
	Independent non-invariant & 0.954 & 0.002 & 0.0195 & 0.021 \\
	Dependent invariant & 0.925 & 0.071 & 0.813 & 0.050 \\
	Dependent non-invariant & 0.815 & 0.137 & 0.0167 & 0.160 \\
	\hline
	\end{tabular}
\end{table}

\begin{table*}[t]
	\centering
	\normalsize
	\renewcommand{\arraystretch}{1.}
	\caption{Recursive multi-step comparison for the local five-step experiment.}
	\label{tab:recursive_multistep_summary}
	\begin{tabular}{lccccc}
	\hline
	\textbf{Method} & \textbf{Final cov.} & \textbf{Pathwise cov.} & \textbf{Avg. radius} & \textbf{Avg. volume} & \textbf{Quantiles} \\
	 & \textbf{at $H_{\rm loc}$} & & & & \textbf{calibrated} \\
	\hline
	Naive one-step CP~\cite{vovk2012conditional} & 0.545 & 0.352 & 0.129 & 1.02 & 1 \\
	imCP & 0.894 & 0.891 & 0.386 & 11.3 & 1 \\
	Fixed-horizon CP~\cite{lindemann2023safe} & 0.975 & 0.887 & 2.327 & 21 & 5 \\
	Trajectory max CP~\cite{cleaveland2024conformal} & 0.994 & 0.975 & 3.014 & 28.5 & 1 \\
	PID~\cite{angelopoulos2023pid} & 0.975 & 0.906 & 2.525 & 22.6 & 5 \\
	EnbPI~\cite{xu2021enbpi} & 0.975 & 0.890 & 2.494 & 22.7 & 5 \\
	\hline
	\end{tabular}
\end{table*}
This system combines a deterministic rotational component with a stable linear subsystem subject to state-dependent disturbances. Such dynamics naturally arise in applications where disturbances are triggered by specific configurations or phases, such as rotating machinery, robotic contact events, or periodic processes.

\begin{figure}[bt]
	\centering
	\includegraphics[scale=.5,trim={0cm 0cm 0cm 0cm},clip]{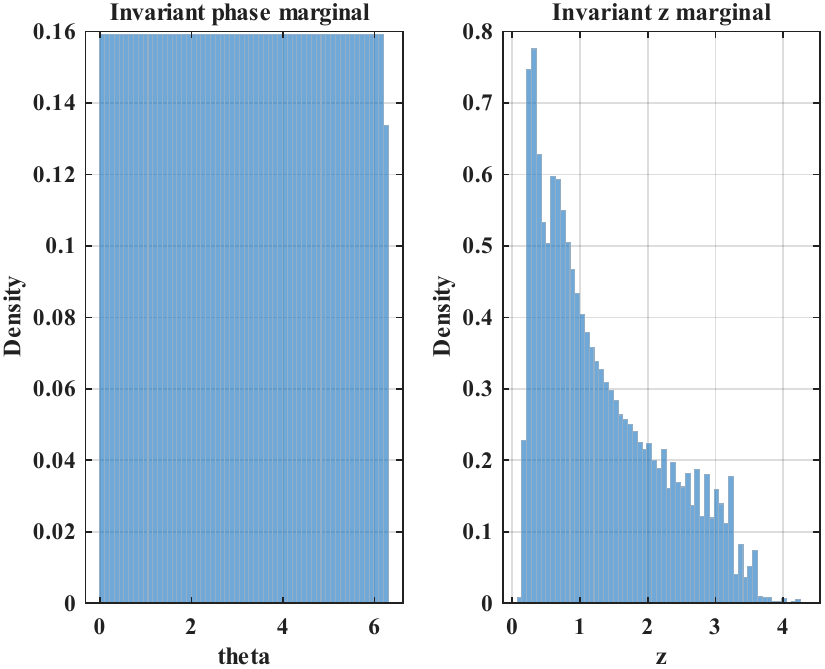}
	\caption{Numerical approximation of the invariant distribution for the rotational stochastic benchmark system. The measure is nearly uniform in angle and concentrated near the stable $z$-region.}
	\label{fig:invariant_estimation_app}
\end{figure}

The angular dynamics are ergodic and induce an approximately uniform invariant distribution over $\theta$. Combined with the stable $z$-dynamics, this yields a stationary invariant measure over the full state space. We approximate this invariant distribution by simulating a long trajectory and discarding an initial burn-in segment. The remaining samples form an empirical approximation of the invariant measure. Fig.~\ref{fig:invariant_estimation_app} illustrates this invariant distribution. The distribution is approximately uniform in $\theta$ and concentrated around $z=0$ due to the stability of the $z$-dynamics.

As in the previous setup, we assume that the true dynamics are unknown and only state transitions are observed. Based on these observations, we construct a predictor $\hat f$ that models the dynamics as a simple linear autoregressive model. This predictor is misspecified, since it does not capture the true localized angular structure of the dynamics, but it serves as a realistic example of a learned model that may be used for prediction and uncertainty quantification.

In the rolling one-step setting, the predictor is reinitialized at every time step using the true observed state. For each random split \(b\), we construct calibration pairs \(\{(X_{i,b}^{\mathrm{cal}},X_{i,+,b}^{\mathrm{cal}})\}_{i=1}^{N}\) with \(X_{i,b}^{\mathrm{cal}}\sim\mu\) and \(X_{i,+,b}^{\mathrm{cal}}\sim P(X_{i,b}^{\mathrm{cal}},\cdot)\). The calibration scores are
\[
S_{i,b}^{\mathrm{cal}} = s_{\hat f}(X_{i,b}^{\mathrm{cal}},X_{i,+,b}^{\mathrm{cal}}),
\qquad i=1,\dots,N,
\]
and the split conformal radius is the finite-sample order statistic
\[
\hat q_{1-\alpha}^{(b)}
= S_{(\lceil (N+1)(1-\alpha)\rceil),b}^{\mathrm{cal}},
\]
where \(S_{(j),b}^{\mathrm{cal}}\) denotes the \(j\)-th ordered calibration score. This radius defines the rolling one-step set \(C_{\alpha,b}^{(1)}(x_k)=\{\xi:s_{\hat f}(x_k,\xi)\le \hat q_{1-\alpha}^{(b)}\}\), which is tested on independent trajectories initialized from the invariant measure.

To study the effect of the calibration distribution, we compare four calibration schemes while keeping the deployment setting fixed. imCP uses independent calibration states sampled from the estimated invariant measure. The independent non-invariant baseline samples calibration states from a truncated Gaussian away from the phase region where the disturbance shock occurs, and then simulates one successor from each state. The dependent invariant baseline uses consecutive post-burn-in trajectory pairs. The dependent non-invariant baseline uses a short transient trajectory initialized outside the support of the invariant probability measure. For every method, the empirical rolling coverage is computed as
\begin{equation}
{\color{black}
\widehat{\operatorname{Cov}}_{\rm roll}
=\frac{1}{N_{\rm split}N_{\rm test}}
\sum_{b=1}^{N_{\rm split}}\sum_{i=1}^{N_{\rm test}}
\mathbf{1}\{x_{b,i,+}^{\rm test}\in C_{\alpha,b}^{(1)}(x_{b,i}^{\rm test})\},
}
\label{eq:rolling_coverage_app}
\end{equation}
where \(b\) indexes the random split and \(\{(x_{b,i}^{\rm test},x_{b,i,+}^{\rm test})\}_{i=1}^{N_{\rm test}}\) denotes the rolling test pairs generated for that split.

The one-step calibration-law comparison is summarized in Table~\ref{tab:one_step_calibration_laws}. The Calibration scheme column specifies both the sampling law and the dependence structure used to form the calibration scores. Coverage is the empirical rolling one-step coverage averaged over all test trajectories, time steps, and random splits. Std. is the standard deviation of this coverage across random splits, where larger values indicate a less stable calibration procedure. Avg. radius is the mean conformal radius $\hat q_{1-\alpha}$ used to form the one-step prediction balls. Gap is defined as $|\widehat{\operatorname{Cov}}-(1-\alpha)|$, so smaller values indicate closer agreement with the target coverage level.

For local recursive multi-step prediction, we use $H_{\rm loc}=5$ steps. All multi-step calibration states are sampled independently from the same estimated invariant measure used above, and successor trajectories are generated by rolling out the true stochastic dynamics from those initial states. Thus, the calibration law remains aligned with the invariant deployment regime; what changes across methods is how calibrated errors are converted into horizon-dependent radii. Starting from an invariant initial state, the predictor is then applied recursively without reinitialization.

Under this common setup, we compare six constructions. The naive baseline reuses the one-step conformal radius at every horizon. imCP uses the same one-step conformal quantile, but propagates it through the Lipschitz factor from the multi-step theorem in the main paper. Fixed-horizon conformal prediction calibrates a separate radius for each horizon, while trajectory-max conformal prediction calibrates the maximum prediction error over the five-step horizon and reuses that radius at every step. The final two baselines are online conformal time-series methods: a PID conformal controller~\cite{angelopoulos2023pid} and an EnbPI-style rolling residual method~\cite{xu2021enbpi}. Both are applied horizon-wise and initialized from the same invariant calibration errors; PID updates after each observed test residual, whereas EnbPI updates its residual buffer causally after each test batch.
The corresponding local five-step comparison is summarized in Table~\ref{tab:recursive_multistep_summary}.


\section{Appendix E: Proofs}

This appendix contains the proofs omitted from the main submission. The theorem, lemma, and equation numbering below follows the local appendix numbering.
\subsection{Proof 1}

\begin{proof}
	By assumption, $x_0 \sim \mu$. Since $\mu$ is invariant for $P$, it follows that for all $k \ge 0$,
	\begin{equation}
	x_k \sim \mu.
	\label{eq:auto-display-014}
	\end{equation}
	Moreover, the realized successor $x_{k+1}$ is generated from $P(x_k,\cdot)$, just as each calibration successor $X_{i,+}^{\mathrm{cal}}$ is generated from $P(X_i^{\mathrm{cal}},\cdot)$. Therefore, for each fixed $k$, the test pair $(x_k,x_{k+1})$ has the same one-step stationary law as the calibration pairs
	\begin{equation}
	\bigl(X_i^{\mathrm{cal}},X_{i,+}^{\mathrm{cal}}\bigr), \qquad i=1,\dots,N.
	\label{eq:auto-display-015}
	\end{equation}
	Since the calibration pairs are independent of the test trajectory, the scores
	\begin{equation}
	S_1^{\mathrm{cal}},\dots,S_N^{\mathrm{cal}},S_k
	\label{eq:auto-display-016}
	\end{equation}
	are exchangeable for every fixed $k$.
	
	By the standard split conformal argument based on order statistics, the conformal quantile $\hat q_{1-\alpha}$ satisfies
	\begin{equation}
	\mathbb{P}\bigl(S_k\le \hat q_{1-\alpha}\bigr)\ge 1-\alpha.
	\label{eq:auto-display-017}
	\end{equation}
	The equivalent set-membership statement follows from the definition
	\begin{equation}
	\mathcal{C}^{(1)}_{\alpha}(x)
	=
	\left\{
	\xi\in\mathcal{X}:
	s_{\hat f}(x,\xi)\le \hat q_{1-\alpha}
	\right\}.
	\label{eq:auto-display-018}
	\end{equation}
\end{proof}

\subsection{Lemma: Error recursion}

\begin{proof}
	Starting from the definition of the error at step $h+1$,
	\begin{equation}
	e_{h+1} = x_{h+1} - \hat x_{h+1} = f(x_h)+n_h - \hat f(\hat x_h).
	\label{eq:auto-display-022}
	\end{equation}
	Adding and subtracting $\hat f(x_h)$ gives
	\begin{equation}
	\|e_{h+1}\|
	=
	\|x_{h+1} - \hat f(x_h) + \hat f(x_h) - \hat f(\hat x_h)\|.
	\label{eq:auto-display-023}
	\end{equation}
	Applying the triangle inequality,
	\begin{equation}
	\|e_{h+1}\|
	\le
	\|x_{h+1} - \hat f(x_h)\|
	+
	\|\hat f(x_h) - \hat f(\hat x_h)\|.
	\label{eq:auto-display-024}
	\end{equation}
	The first term equals the one-step score $S_h$, and the second term is bounded by the Lipschitz continuity of $\hat f$:
	\begin{equation}
	\|\hat f(x_h) - \hat f(\hat x_h)\|
	\le
	L\|x_h-\hat x_h\|
	=
	L\|e_h\|.
	\label{eq:auto-display-025}
	\end{equation}
	Combining these bounds yields the desired result.
\end{proof}

\subsection{Lemma: Multi-step error expansion}

\begin{proof}
	We prove the result by repeatedly applying the error-recursion lemma in the main paper. For $k=1$, the result follows directly:
	\begin{equation}
	\|e_1\| \le S_0.
	\label{eq:auto-display-026}
	\end{equation}
	For $k=2$,
	\begin{equation}
	\|e_2\| \le L\|e_1\|+S_1 \le LS_0+S_1.
	\label{eq:auto-display-027}
	\end{equation}
	Proceeding inductively, each application of the recursion multiplies the previous terms by $L$ and adds a new score term. After $k$ steps, this yields
	\begin{equation}
	\|e_k\|
	\le
	L^{k-1}S_0 + L^{k-2}S_1+\cdots+S_{k-1}
	=
	\sum_{i=0}^{k-1}L^{k-1-i}S_i.
	\label{eq:auto-display-028}
	\end{equation}
\end{proof}

\subsection{Lemma: $k$-step probability bound}

\begin{proof}
	Using the union bound,
	\begin{equation}
	\mathbb{P}\left(\bigcup_{i=0}^{k-1} A_i^c\right)
	\le
	\sum_{i=0}^{k-1}\mathbb{P}(A_i^c).
	\label{eq:auto-display-036}
	\end{equation}
	From the one-step conformal guarantee,
	\begin{equation}
	\mathbb{P}(S_i > \hat q_{1-\alpha}) \le \alpha,
		\qquad \forall i=0,\dots,k-1.
	\label{eq:auto-display-037}
	\end{equation}
	Therefore,
	\begin{equation}
	\mathbb{P}\left(\bigcup_{i=0}^{k-1} A_i^c\right)
	\le
	k\alpha.
	\label{eq:auto-display-038}
	\end{equation}
	Taking complements yields
	\begin{equation}
	\mathbb{P}\left(\bigcap_{i=0}^{k-1} A_i\right)
	\ge
	1-k\alpha.
	\label{eq:auto-display-039}
	\end{equation}
	Since probabilities are nonnegative, this can be written as
	\begin{equation}
	\mathbb{P}\left(\bigcap_{i=0}^{k-1} A_i\right)
	\ge
	1-\min\{1,k\alpha\}.
	\label{eq:auto-display-040}
	\end{equation}
\end{proof}

\subsection{Theorem: Multi-step conformal guarantee}

\begin{proof}
	Define the event
	\begin{equation}
	E := \bigcap_{i=0}^{k-1}\{S_i\le \hat q_{1-\alpha}\}.
	\label{eq:auto-display-041}
	\end{equation}
	By the $k$-step probability bound in the main paper,
	\begin{equation}
	\mathbb{P}(E)
	\ge
	1-\min\{1,k\alpha\}.
	\label{eq:auto-display-042}
	\end{equation}
	On the event $E$, all one-step scores satisfy
	\begin{equation}
	S_i \le \hat q_{1-\alpha}, \qquad i=0,\dots,k-1.
	\label{eq:auto-display-043}
	\end{equation}
	Substituting this bound into the multi-step error expansion for each $j=1,\dots,k$ yields
	\begin{equation}
	\|x_j-\hat x_j\|
	\le
	\hat q_{1-\alpha}\sum_{i=0}^{j-1}L^i.
	\label{eq:auto-display-044}
	\end{equation}
	Using the closed-form expression for the geometric sum, we obtain
	\begin{equation}
	\|x_j-\hat x_j\|
	\le
	\hat q_{1-\alpha}\Gamma_j(L),
	\qquad j=1,\dots,k.
	\label{eq:auto-display-045}
	\end{equation}
	Therefore,
	\begin{align}
	&\mathbb{P}\left(
	\bigcap_{j=1}^{k}
	\left\{
	\|x_j-\hat x_j\|
	\le
	\hat q_{1-\alpha}\Gamma_j(L)
	\right\}
	\right) \notag\\
	&\qquad\ge
	\mathbb{P}(E)
	\ge
	1-\min\{1,k\alpha\}.
	\label{eq:auto-display-046}
	\end{align}
	This proves the recursive multi-step pathwise guarantee. Since the terminal event is implied by the pathwise event, the terminal-horizon guarantee follows immediately.
\end{proof}
\ifdefined\incombinedpaper
\def\maybeenddocument{}
\else
\bibliography{references}
\def\maybeenddocument{\end{document}}
\fi
\maybeenddocument

\bibliography{references}
\end{document}